\documentclass[
final
]{dmtcs-episciences}        
\usepackage{graphics, amsthm, amsmath, amssymb, algorithm, algorithmic}
\newtheorem{theorem}{Theorem}

\newtheorem{lemma}{Lemma}
\newtheorem{corollary}{Corollary}

\newtheorem{property}{Property}

\author{Jessica Enright\affiliationmark{1}
	\and Lorna Stewart\affiliationmark{2}}
\title{Equivalence of the filament and overlap graphs of subtrees of limited trees}

\affiliation{
	University of Stirling, Stirling, Scotland, UK\\
	University of Alberta, Edmonton, Alberta, Canada}
	
	\received{2015-8-7}
	
	\revised{2016-7-27, 2017-6-2}
	
	\accepted{2017-6-2}

\begin{document}
	\publicationdetails{19}{2017}{1}{20}{1274}

\maketitle

\begin{abstract}
The overlap graphs of subtrees of a tree are equivalent to subtree filament graphs, the overlap graphs of subtrees of a star are cocomparability graphs, and the overlap graphs of subtrees of a caterpillar are interval filament graphs. In this paper, we show the equivalence of many more classes of subtree overlap and subtree filament graphs, and equate them to classes of complements of cochordal-mixed graphs. 
Our results generalise the previously known results mentioned above.
\end{abstract}

\section{Introduction}\label{intro}
The class of subtree overlap graphs is equivalent to the class of subtree filament graphs, which means that those graphs have both overlap and filament intersection representations on trees \cite{Jess}.
The class contains many graph classes that have extensive structural properties, algorithms, and complexity results, such as interval graphs, permutation graphs, cocomparability graphs, chordal graphs, circle graphs, circular-arc graphs, polygon-circle graphs, and interval filament graphs.  
Some of these graphs have been characterised in terms of subtree overlap representations on restricted host trees. In particular, cocomparability graphs are the overlap graphs of subtrees of a star (follows from \cite{GolSch}), circle graphs are the overlap graphs of subtrees of a path \cite{gavril1973}, and interval filament graphs are the overlap graphs of subtrees of a caterpillar (this fact was presented at a workshop but not published \cite{CGO}).
Thus, we have the equivalence of general subtree overlap graphs and subtree filament graphs, and we know that some subtree overlap graphs that admit representations on restricted host trees are equivalent to well-known graph classes, one of which is a natural class of subtree filament graphs.

In this paper, we identify new equivalences between subtree overlap and subtree filament graph classes based on host trees of their representations.
We first introduce the notion of a covering subtree of a tree representation. We show that the host tree of any subtree overlap representation can be modified so that it consists of just a covering subtree plus some additional leaves, without altering the represented graph.
In addition, we prove that a graph has a subtree overlap representation with a given covering subtree if and only if it is the complement of a restricted type of cochordal-mixed graph. 
Finally, we show that for a set $\mathcal{S}$ of trees that is closed under edge subdivision, a graph has a subtree filament representation with host tree in $\mathcal{S}$ if and only if it has a subtree overlap representation with covering subtree in $\mathcal{S}$. 

Our first theorem generalises the characterisation of cocomparability graphs as the overlap graphs of subtrees of a star by equating the overlap graphs of subtrees of a star with subtree overlap graphs that have a representation with a single-vertex covering subtree. Theorem 1 generalises this correspondence for any given covering subtree.

Our second theorem bridges the gap between the previously known equivalence of general subtree overlap graphs and subtree filament graphs, and the characterisation of interval filament graphs as the overlap graphs of subtrees of a caterpillar. Both of those previously known results and new equivalences between subtree overlap and subtree filament graph classes are given in Theorem 2. The theorem suggests a division of subtree overlap graphs into subclasses,
each of which contains all interval filament graphs, and the union of which
is the class of subtree overlap graphs. 
This view of subtree overlap graphs may give insight into their structure and the algorithmic complexity of problems in that domain. 

We consider finite, simple graphs. Let $G=(V,E)$ be a graph.
The {\em neighbourhood} of a vertex $v \in V$ in $G$ is $N_G(v) = \{ u ~|~ uv \in E \}$.
$K_n$ denotes the complete graph on $n$ vertices.

Two sets $A$ and $B$ {\em intersect} if $A \cap B \ne \emptyset$, and
{\em overlap}, denoted $A \between B$, if $A \cap B \ne \emptyset$, $A \not\subseteq B$, and $B \not\subseteq A$. Sets $A$ and $B$ are {\em disjoint}, denoted $A | B$, if $A \cap B = \emptyset$.
Let $A, B, A', B'$ be four nonempty sets.  We say that $(A, B)$ and $(A', B')$ are \emph{similarly related}, denoted $(A, B) \sim (A',B')$ if $A|B$ if and only if $A'|B'$ and $A \between B$ if and only if $A' \between B'$. 

Let $\mathcal{S} = \{ S_1, S_2, \ldots, S_n \}$ be a multiset of nonempty sets. We use the term multiset rather than set to allow for the possibility that $S_i = S_j$ for some $1 \le i,j \le n$ where $i \ne j$. The {\em intersection graph} 
(respectively, {\em overlap graph}, {\em disjointness graph}, {\em containment graph}) of $\mathcal{S}$ is the graph $G=(V,E)$ where $V = \{ v_1, v_2, \ldots, v_n\}$ and, for all $1 \le i,j \le n$, $v_i v_j \in E$ if and only if $S_i$ and $S_j$ intersect (respectively, overlap, are disjoint, are contained one in the other). 
If $G$ is the intersection,
overlap, disjointness, or containment graph of $\mathcal{S}$ then $\mathcal{S}$ is
called an intersection, overlap, disjointness, or containment representation
of $G$. Every graph has both an intersection and a disjointness
representation \cite{Marcz} as well as an overlap representation (obtained
by adding a unique new element to each set of
an intersection representation). 
Note
that, for $\mathcal{S} = \{ S_1, S_2, \ldots, S_n \}$ and 
$\mathcal{S'} = \{ S'_1, S'_2, \ldots, S'_n \}$ where
$(S_i, S_j) \sim (S'_i, S'_j)$ for all $1 \le i,j \le n$,
the intersection (respectively,
overlap, disjointness, containment) graphs of $\mathcal{S}$ and $\mathcal{S'}$ are
isomorphic.

In this paper, we are concerned with the intersection, overlap, disjointness, and containment graphs of subtrees of a tree. For a given tree $T$, we will assume that a collection of $n$ subtrees of $T$ is given as a multiset $\mathcal T = \{ t_1, \ldots, t_n \}$ of subsets of the vertices of $T$, each of which induces a subtree of $T$.

{\em Interval graphs} are the intersection graphs of intervals
on a line or, equivalently, the intersection graphs of subtrees of a path.
{\em Cointerval graphs} are the complements of interval graphs, that is, the disjointness graphs of subtrees of a path.
{\em Circle graphs} are the 
intersection graphs of chords in a circle or, equivalently, the overlap graphs of subtrees of a path.  {\em Chordal graphs} are graphs in which
every cycle of length greater than three has a chord
or, equivalently, the intersection graphs of subtrees in
a tree. 
{\em Cochordal graphs} are the complements of chordal graphs.
{\em Comparability graphs} are graphs whose edges
can be transitively oriented. Equivalently, comparability
graphs are the containment graphs of subtrees of a
tree, the containment graphs of subtrees of a star, and
the set of all containment graphs \cite{GolSch}. 
{\em Cocomparability} graphs are the complements of comparability graphs.
{\em Subtree overlap graphs} are the overlap graphs of subtrees in a tree.
If a graph $G=(V,E)$ is the overlap (respectively, intersection, containment, or disjointness) graph of 
subtrees $\mathcal{T}$ of a tree $T$, then $\mathcal{T}$ is a \emph{subtree overlap (respectively, intersection, containment, or disjointness) representation} of $G$. $T$ is termed the {\em host tree} of the representation. For convenience, we will use the notation that vertex $v_i \in V$ corresponds to subtree $t_i \in \mathcal{T}$. 
A {\em caterpillar} is a tree such that the removal of its leaves results in a path.
All of the graph classes defined in this section are {\em hereditary}, that is, every induced subgraph of a graph in the class is also in the class.
For more information about graph classes, see \cite{BLS}.

Interval filament graphs, subtree filament graphs, and $\mathcal G$-mixed graphs were introduced by Gavril \cite{gavril2000}. We give the definitions of those graph classes and related concepts next.

Let $\mathcal I = \{I_1, \ldots, I_n \}$ be a multiset of (closed) intervals on a line $L$ and let $P$ be a plane containing $L$. We will refer to one of the half-planes into which $L$ divides $P$ as being above $L$.
Then $\mathcal F = \{f_1, \ldots, f_n \}$ is a set of interval filaments on the intervals of $\mathcal I$ if
each $f_i$, $1 \le i \le n$, is a curve in $P$, on and above $L$, connecting the endpoints of $I_i$, such that if two intervals are disjoint then their curves do not intersect.
Thus, pairs of filaments corresponding to disjoint intervals do not intersect, pairs of filaments corresponding to overlapping intervals intersect, and pairs of filaments corresponding to intervals where one is contained in the other may or may not intersect.
The intersection graph of $\mathcal F$ is called an {\em interval filament graph}.

Subtree filaments and subtree filament graphs are defined analogously.
Let $T$ be a tree and let $\mathcal T = \{ t_1, \ldots, t_n \}$ be a multiset of subtrees of $T$. Suppose that $T$ is embedded in a plane $P$ and let $S$ be a surface perpendicular to $P$ whose intersection with $P$ is $T$. 
(One can imagine forming the part of $S$ that is above $T$ by drawing $T$ upwards from $P$.)
Then $\mathcal F = \{f_1, \ldots, f_n \}$ is a set of subtree filaments on the subtrees of $\mathcal T$ if
each $f_i$, $1 \le i \le n$, is a curve in $S$, in and above $T$, connecting the leaves of $t_i$, such that 
(i) if two subtrees are disjoint then their curves do not intersect, and
(ii) if two subtrees overlap then their curves intersect.
Thus, pairs of filaments corresponding to disjoint subtrees do not intersect, pairs of filaments corresponding to overlapping subtrees intersect, and 
pairs of filaments corresponding to intervals where one is contained in the other may or may not intersect.
If a graph $G$ is the intersection graph of a collection of filaments on subtrees of a tree $T$, then $G$ is a {\em subtree filament graph} and the collection of filaments is a {\em subtree filament representation} of $G$. The tree $T$ is called the {\em host tree} of the representation. 

Let $\mathcal{G}$ a hereditary graph class. A graph $G = (V, E)$ is said to be {\em $\mathcal{G}$-mixed} if there is a partition of its edges into $E_1$ and $E_2$ such that:
\begin{itemize}
 \item $G_1 = (V, E_1)$ is in $\mathcal{G}$ and 
 \item there is a transitive orientation 
$ (V, \overrightarrow{E_2})$ 
 of the graph $(V, E_2)$ such that
 for every three distinct vertices $u,v,w \in V$, if
      $(u \rightarrow v) \in \overrightarrow{E_2}$ and $v w \in E_1$, then $u w \in E_1$.
 \end{itemize}
Such a partition is called a {\em $\mathcal{G}$-mixed partition} of the edges of $G$.
Subtree filament graphs are exactly the complements of cochordal-mixed graphs, and interval filament graphs are exactly the complements of cointerval-mixed graphs \cite{gavril2000}.
 
Let $\mathcal{T}$ be a multiset of subtrees of a tree $T$.
A subtree $t$ of $T$ is called
a {\em covering subtree} of 
$\mathcal{T}$ 
if it intersects every member of $\mathcal{T}$. 
Note that the intersection of each element of $\mathcal{T}$ with a covering subtree $t$ is a subtree of $t$.
Let $T$ be a tree and let $t$ be a subtree of $T$.
A vertex $v$ of $t$ is called {\em bushy} 
(with respect to $t$ in $T$)
if every neighbour of $v$ that is not in $t$ is a leaf of $T$; the entire subtree $t$ is called {\em bushy} 
(in $T$) 
if every vertex of $t$ is bushy
(with respect to $t$ in $T$).  An example of a tree that contains both bushy and non-bushy vertices is given in Figure \ref{fig:bushinessExample}.
\begin{figure}
\begin{center}
\scalebox{1}{\includegraphics{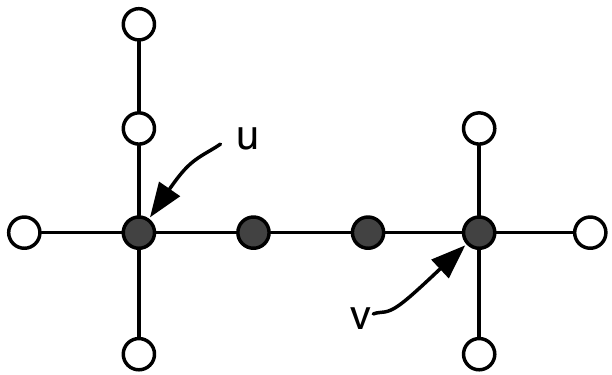}}
\caption{A tree $T$ and a subtree $t$, induced by the darker grey vertices.  The vertex $u$ is not bushy with respect to $t$ in $T$ because it has a neighbour (above it in the diagram) that is not in $t$ and is not a leaf of $T$.  Vertex $v$ is bushy with respect to $t$ in $T$ because all of its neighbours that are not in $t$ are leaves of $T$. }
\label{fig:bushinessExample}
\end{center}
\end{figure}
Now we can define the graph classes that will be examined in Section \ref{equiv}.
Let $G$ be a graph and $\mathcal{S}$ be a set of trees. 
\begin{itemize}
\item
$G$ is an {\em $\mathcal{S}$-covered subtree overlap graph} if it has a subtree overlap representation that has a covering subtree isomorphic to a tree in $\mathcal{S}$.
Such a representation is an {\em $\mathcal{S}$-covered subtree overlap representation} of $G$. 
\item
$G$ is a {\em bushy $\mathcal{S}$-covered subtree overlap graph} if it has a subtree overlap representation with a bushy covering subtree isomorphic to a tree in $\mathcal{S}$.
Such a representation is a {\em bushy $\mathcal{S}$-covered subtree overlap representation} of $G$. 
\item
$G$ is an {\em $\mathcal{S}$-subtree-filament graph} if there is a subtree filament representation of $G$ such that the host tree is isomorphic to a member of $\mathcal{S}$. Such a representation is an {\em $\mathcal{S}$-subtree-filament representation} of $G$.
\item
$G$ is an {\em $\mathcal{S}$-cochordal graph} if it 
has a subtree disjointness representation such that the host tree is isomorphic to a member of $\mathcal{S}$.
Such a representation is an {\em $\mathcal{S}$-cochordal representation} of $G$. 
\end{itemize}
When $\mathcal{S}$ has just one element, say $T$, we sometimes write $T$ instead of $\{T\}$ in the above notation.

\section{Subtree representations}

In this section, we give methods for transforming a given multiset of subtrees of a tree into another representation of the same type 
(i.e., intersection, overlap, disjointness, containment) 
for the same graph.
Let $G=(V,E)$ be a graph and let $x \notin V$. For any given $uv \in E$, 
the {\em subdivision of edge $uv$ (with vertex $x$)} is the operation of removing the edge $uv$ from $G$ and adding the vertex $x$ and the edges $ux$ and $xv$. The vertex $x$ is called a {\em subdivision vertex}. A graph $H$ is a \emph{subdivision} of $G$ if $H$ can be obtained from $G$ by zero or more edge subdivisions.  

We first mention some simple alterations that can be made to a given tree $T$ and multiset $\mathcal T$ of subtrees of $T$ without changing the relationships among the elements of $\mathcal T$.
First, if a new leaf is added to $T$ but to no element of $\mathcal T$, then the subtrees are unchanged and therefore the relationships among them remain the same. Second, if an edge is subdivided in $T$ and in every element of $\mathcal T$ that contains the edge, then the relationships among the subtrees remain unchanged.
Therefore, if $T$ and $T'$ are trees such that $T'$ can be obtained from $T$ by a sequence of leaf additions and edge subdivisions, then 
any graph that is the intersection, overlap, containment, or disjointness graph of subtrees of $T$ is also the 
intersection, overlap, containment, or disjointness (respectively) graph of subtrees of $T'$.
Furthermore,
the $T$-subtree-filament graphs form a subset of the $T'$-subtree filament graphs, and the $T$-covered subtree overlap graphs form a subset of the $T'$-covered subtree overlap graphs.

Recalling that we denote subtrees as subsets of vertices of a tree, each of which induces a subtree:
\begin{lemma} \label{lem:subdiv}
Let $T=(V_T, E_T)$ be a tree and $\mathcal{T} = \{ t_1, \ldots, t_n  \}$ be a multiset of subtrees of $T$. 
Let $v$ be a vertex of $T$ and let $w$ be a neighbour of $v$ in $T$.
Let $T'$ be the tree obtained from $T$ by 
subdividing the edge $vw \in E_T$ with a vertex $x \notin V_T$.
Let $\mathcal S$ be a (possibly empty) subset of $\mathcal T$, where each element of 
$\mathcal S$ contains $v$.  
For all $1 \le i \le n$, let
\[
t'_i = \left\{
\begin{array}{ll}
t_i \cup \{x\} & \mbox{if } v, w \in t_i, \mbox{ or } t_i \in \mathcal S, \mbox{ or there is an element } t_k \in \mathcal S \mbox{ such that } t_k \subset t_i,\\
t_i & \mbox{otherwise.}
\end{array}
\right.
\]
Then
$\{ t'_1, \ldots, t'_n \}$ is a multiset of subtrees of $T'$ and,
for all $1 \le i,j \le n$, $(t_i, t_j) \sim (t'_i, t'_j)$.
\end{lemma}

\begin{proof}
The elements of $\{ t'_1, \ldots, t'_n \}$ induce connected subgraphs of $T'$ by the construction.
Note that, for all $1 \le i \le n$, 
$t'_i = t_i \cup \{x\}$ only if $v \in t_i$, and
$t'_i = t_i$ if and only if $x \notin t'_i$.
Let $t_i, t_j \in \mathcal T$ where $1 \le i,j \le n$.
We will prove that $(t_i, t_j) \sim (t'_i, t'_j)$ by showing that
both of the following hold:
$t_i \cap t_j = \emptyset$ if and only if $t'_i \cap t'_j = \emptyset$;  ($t_i \backslash t_j = \emptyset$ or $t_j \backslash t_i = \emptyset$) if and only if ($t'_i \backslash t'_j = \emptyset$ or $t'_j \backslash t'_i = \emptyset$). There are three cases to be considered, based on whether $x$ is in neither, one, or both of $t'_i$ and $t'_j$.

If $x$ is in neither $t'_i$ nor $t'_j$ then $t'_i = t_i$ and $t'_j = t_j$ and the result clearly follows.

If $x$ is in both $t'_i$ and $t'_j$ then $v$ is in both $t_i$ and $ t_j$, and therefore $t_i \cap t_j$ and $t'_i \cap t'_j$ are both nonempty.
Furthermore, $t'_i \backslash t'_j = (t_i \cup \{x\}) \backslash (t_j \cup \{x\}) = t_i \backslash t_j$ and, similarly, $t'_j \backslash t'_i = t_j \backslash t_i$. So the result follows.

If $x$ is in just one of $t'_i$ and $t'_j$, 
then suppose without loss of generality that $x$ is in $t'_i$ and not in $t'_j$. 
Then $t'_i = t_i \cup \{x\}$, $t'_j = t_j$, and $v \in t_i$. 
So 
$t'_i \cap t'_j = (t_i \cup \{x\}) \cap t_j = t_i \cap t_j$.
Furthermore,
$t'_i \backslash t'_j = (t_i \cup \{x\}) \backslash t_j = (t_i \backslash t_j) \cup \{x\} \ne \emptyset$,
and $t'_j \backslash t'_i = t_j \backslash (t_i \cup \{x\}) = t_j \backslash t_i$.
If $t_j \backslash t_i = \emptyset$ then $t'_j \backslash t'_i = \emptyset$ and we have 
($t_i \backslash t_j = \emptyset$ or $t_j \backslash t_i = \emptyset$) if and only if ($t'_i \backslash t'_j = \emptyset$ or $t'_j \backslash t'_i = \emptyset$),
so the result follows.
In the remainder of the proof, we handle the case where $t_j \backslash t_i \ne \emptyset$.
Since $t'_j \backslash t'_i = t_j \backslash t_i$ we also have $t'_j \backslash t'_i \ne \emptyset$ and, from before, $t'_i \backslash t'_j \ne \emptyset$.
Therefore, to complete the proof, we must show that $t_i \backslash t_j \ne \emptyset$.
Suppose for contradiction that $t_i \backslash t_j = \emptyset$, that is, $t_i \subseteq t_j$, which implies that $t_i \subset t_j$ since $t_j \backslash t_i \ne \emptyset$.
Since $x$ is in $t'_i$, one of the following must hold: 
(i) $v,w \in t_i$,
(ii) $t_i \in \mathcal S$, 
or 
(iii) there is an element $t_k \in \mathcal S$ such that $t_k \subset t_i$.
But then, since $t_i \subset t_j$, one of the following must also hold (respectively):
(i) $v, w \in t_j$, 
(ii) $t_i$ is an element of $\mathcal S$ such that $t_i \subset t_j$, 
or
(iii) $t_k$ is an element of $\mathcal S$ such that $t_k \subset t_i \subset t_j$.
In any case, this implies that $x \in t'_j$, a contradiction.
Therefore $t_i \backslash t_j \ne \emptyset$ as required.

In each case, $(t_i , t_j) \sim (t'_i , t'_j)$ and therefore the proof is complete.
\end{proof}

In order to define a subtree filament representation on subtrees $\mathcal T$ of a tree $T$, it is convenient to make some assumptions about the elements of $\mathcal T$. Since we are concerned with the host trees of representations, we need to consider the effect on a host tree of enforcing those assumptions. This is the subject of the next definition and lemma, which will be used in the proof of Theorem \ref{bigThm}.

\begin{property}\label{A}
Subtrees
$\mathcal{T}$ of tree $T$ are said to satisfy the {\em nontrivial intersection distinct leaf property} if:
\begin{itemize}
\item
each element of $\mathcal{T}$ is nontrivial, 
\item
every pair of elements of $\mathcal{T}$ are either disjoint or share two or more vertices, and 
\item
no vertex of $T$ is a leaf of two distinct members of $ \mathcal{T}$.
\end{itemize}
\end{property}

Note that the third requirement in the above property guarantees that the elements of $\mathcal T$ are distinct.

The next lemma shows that every multiset of subtrees of a tree can be transformed into subtrees of a tree that satisfy Property \ref{A}, without altering the relationships among the subtrees.
Because we are concerned with specific host trees, we consider the effect of the transformation on the host tree.

\begin{lemma} \label{lem:bushyGivesNice}
Let $T=(V_T, E_T)$ be a nontrivial tree and $\mathcal{T} = \{ t_1, \ldots, t_n  \}$ be a multiset of subtrees of $T$.
There exists a tree $T'$ and set $\mathcal{T}' = \{ t'_1, \ldots, t'_n \}$ of subtrees of $T'$ such that
\begin{itemize}
\item
for all $1 \le i,j \le n$, $(t_i, t_j) \sim (t'_i, t'_j)$,
\item
$T'$ is a subdivision of $T$,
and
\item
$\mathcal{T}'$ satisfies Property \ref{A}.
\end{itemize}
\end{lemma}

\begin{proof}
Let $T=(V_T, E_T)$ be a nontrivial tree and $\mathcal{T} = \{ t_1, \ldots, t_n  \}$ be a multiset of subtrees of $T$.
We may assume that no leaf of $T$ is contained in any element of $\mathcal{T}$. Otherwise, for each leaf $\ell$ of $T$ that is contained in an element of $\mathcal{T}$, we could add a new leaf to $T$ adjacent to $\ell$. Then $T$ would be isomorphic to a subdivision of the original tree and would satisfy the assumption. 
We show how to transform $T$ and $\mathcal{T}$ into $T'$ and $\mathcal{T}'$, respectively,  such that the conditions of the lemma are satisfied.

First, 
we perform $2 |E_T|$ applications of the transformation of Lemma \ref{lem:subdiv}, as follows.
The initial application is performed on $T$ and $\mathcal T$ 
with $v$ and $w$ being the endpoints of an edge of $T$ and $\mathcal S$ being the set of all elements of $\mathcal T$ that contain $v$. By Lemma \ref{lem:subdiv}, this transformation
results in a tree and a multiset of subtrees of the tree. Each subsequent application is performed on the tree and subtrees resulting from the previous step, and again produces a tree and a multiset of subtrees. In the following description, we refer to the current tree and subtrees as $T_c$ and $\mathcal T_c$ respectively.
Overall,
for each edge $pq$ of the initial tree $T$, we apply the transformation of Lemma \ref{lem:subdiv} twice, once with $v=p$, $w=q$, and $\mathcal S$ being all elements of $\mathcal T_c$ that contain $p$;
then with $v=q$, $w=x$ where $x$ is the subdivision vertex from the previous step, and
$\mathcal S$ being all elements of $\mathcal T_c$ that contain $q$.
By Lemma \ref{lem:subdiv}, this process finally results in a tree $T_2$ and
subtrees $\mathcal{T}_2= \{ t^2_1, \ldots, t^2_n \}$ of $T_2$, with
$(t_i,t_j) \sim (t^2_i,t^2_j)$ for all $ 1 \le i,j \le n$.
Clearly, $T_2$ is a subdivision of $T$.

Let $t_i$ be a trivial element of $\mathcal{T}$ and let $p$ be the single vertex of $t_i$. Since $T$ is nontrivial, $p$ is incident with an edge of $T$, and the new vertex introduced in the subdivision of that edge with $v=p$ is added to $t_i$ in the construction of $t^2_i$.
Therefore each element of $ \mathcal{T}_2 $ is nontrivial.

Suppose that distinct subtrees $t_i, t_j \in \mathcal{T}$ intersect in just one vertex, $p$. 
Since $T$ is nontrivial, $p$ is incident with an edge of $T$. The new vertex that subdivides that edge is in both $t^2_i$ and $t^2_j$. Therefore every pair of elements of $\mathcal{T}_2$ are either disjoint or share an edge.

Since no leaf of $T$ is contained in an element of $\mathcal{T}$,
the leaves of members of $ \mathcal{T}_2 $
are all vertices of $T_2$ that are not in $V_T$.
Therefore, every vertex that is a leaf of a subtree of 
$ \mathcal{T}_2 $ has degree two in $T_2$.
Furthermore, if two elements of $ \mathcal{T}_2 $ share a leaf, say $p$, then they both contain the neighbour of $p$ that played the role of $v$ during the subdivision when $p$ was introduced, and not the other neighbour as that would contradict $p$ being a leaf of both subtrees.

To complete the proof, we show how to 
reduce the number of vertices of $T_2$ that are leaves of two or more distinct elements of $\mathcal{T}_2$. Applied iteratively, this leads to a representation that satisfies all conditions of the lemma.
Let $p$ be a node in $T_2$
and let $\mathcal{T}_p$ be the elements of $\mathcal{T}_2 $ that contain $p$ as a leaf.
Let $n_p = |\mathcal{T}_p|$ and suppose that $n_p \ge 2$.
By the observation of the preceding paragraph, $p$ has degree two in $T_2$
and every element of $\mathcal{T}_p $ also contains one of $p$'s neighbours and not the other.
Let $q$ and $r$ be the neighbours of $p$, such that every element of $\mathcal{T}_p $ contains $q$ and not $r$.
Let the elements of $\mathcal{T}_p$ be
sorted by nondecreasing size so that each element of $\mathcal{T}_p$ has a position from $1$ to $n_p$ in the sorted list.

We now use $n_p$ new vertices, $s_1, \ldots, s_{n_p}$, which are not vertices of $T_2$, as subdivision vertices in $n_p$ applications of Lemma \ref{lem:subdiv}.
The first application is performed on $T_2$ and $\mathcal T_2$, and produces a tree and subtrees of a tree. Each subsequent application is performed on the tree and subtrees resulting from the previous step. We will refer to the current tree and subtrees as $T_c$ and $\mathcal T_c$, respectively, and the current subtrees corresponding to those of $\mathcal T_p$ as $\mathcal T_{pc}$.
First, apply the transformation of Lemma \ref{lem:subdiv} to $T_2$ and $\mathcal T_2$ with $v=p$, $w=r$, and $\mathcal S$ consisting of the single element of $\mathcal T_p$ having position $n_p$ in the sorted list, using the vertex $s_{n_p}$ as the subdivision vertex.
Then, for each $\ell$ from $n_p-1$ down to 1, apply the transformation of Lemma \ref{lem:subdiv} to $T_c$ and $\mathcal T_c$ with $v=p$ and $w=s_{\ell+1}$ (the subdivision vertex from the previous step), with $\mathcal S$ being the set of all elements of $\mathcal T_{pc}$ with corresponding elements in the sorted list of elements of $\mathcal T_p$ having positions greater than or equal to $\ell$, and using subdivision vertex $s_{\ell}$.
By Lemma \ref{lem:subdiv}, this process terminates with a tree $T_3$ and subtrees $\mathcal T_3$ of $T_3$.

By repeated application of Lemma \ref{lem:subdiv},
$\mathcal T_3$ is a multiset of subtrees of $T_3$ and
$(t^2_i, t^2_j) \sim (t^3_i, t^3_j)$ for all $1 \le i,j \le n$.
In addition,  $T_3$ is a subdivision of $T_2$ and therefore of $T$. Since $t^2_i \subseteq t^3_i$ for all $1 \le i \le n$, each $ t^3_i $ has at least two vertices and each pair of intersecting subtrees of $\mathcal{T}_3$ shares an edge. No vertex has had its degree increased, and only degree two vertices have been added; therefore, every vertex of $T_3$ that is a leaf of any subtree of $ \mathcal{T}_3 $ has degree two in $T_3$. 
By the construction of $\mathcal T_2$, each element $t_i^2$ of $\mathcal T_2 \backslash \mathcal T_p$ contains all or none of $p, q,$ and $r$, and therefore
$t_i^3$ either contains all of the $s_{\ell}$'s and $r$, or none of the $s_{\ell}$'s, and therefore does not contain any of the $s_{\ell}$'s as a leaf.
Each $t_i^3 \in \mathcal T_3$ such that $t_i^2 \in \mathcal T_p$
has exactly one of the $s_{\ell}$'s as a leaf, namely, $s_{\ell}$ where $\ell$ is the position of $t_i^2$ in the sorted list of the elements of $\mathcal T_p$. 
Therefore, $p$ is not a leaf of any element of $ \mathcal{T}_3 $ and each new vertex $s_{\ell}$ of $T_3$ is a leaf of just one element of $ \mathcal{T}_3 $, namely, the subtree corresponding to the element of $\mathcal T_p$ in position $\ell$ of the sorted list. 
Thus the number of vertices of $T_3$ that are leaves of two or more distinct elements of $ \mathcal{T}_3 $ is less than than the number of vertices of $T_2$ that are leaves of two or more distinct elements of $\mathcal{T}_2 $. 
Applied iteratively, this method eventually produces subtrees $\mathcal{T}'$ of a tree $T'$ that satisfy the lemma.
\end{proof}

\section{Equivalence of $\mathcal{S}$-covered subtree overlap graphs and $\mathcal{S}$-subtree-filament graphs}\label{equiv}

In this section, we show that for any tree $R$, 
every $R$-covered subtree overlap graph has a representation in which the host tree is just $R$ with some additional leaves, and that $R$-covered subtree overlap graphs are equivalent to the complements of $R$-cochordal-mixed graphs. This equivalence does not extend to $R$-subtree filament graphs since, for example, 
$C_4$, the cycle on four vertices, is a $K_2$-covered subtree overlap graph but not a $K_2$-subtree-filament graph.
However, the equivalence does extend to subtree filament graphs when edge subdivision is allowed. In Theorem \ref{bigThm}
we show the equivalence of 
$\mathcal{S}$-covered subtree overlap graphs,
$\mathcal{S}$-subtree-filament graphs, and
complements of $\mathcal{S}$-cochordal-mixed graphs, 
when $\mathcal S$ is a nontrivial set of trees that is closed under edge subdivision.

\begin{theorem} \label{th:3parts}
Let $R$ be a tree and let $G$ be a graph. The following statements are equivalent:
\begin{enumerate}
\item $G$ is an $R$-covered subtree overlap graph. \label{3covered}
\item $G$ is the complement of an $R$-cochordal-mixed graph. \label{3mixed}
\item $G$ is a bushy $R$-covered subtree overlap graph. \label{3bushy}
\end{enumerate}
\end{theorem}
\begin{proof}

\textbf{\ref{3covered} $\Rightarrow$ \ref{3mixed}}: 
Let $\mathcal{T} = \{ t_1, \ldots, t_n \}$ be a subtree overlap representation for $G$ in tree $T$ with covering subtree $R$, and suppose that the elements of $\mathcal{T}$ are indexed such that $i<j$ implies $|t_i| \le |t_j|$.
Then the sets $E_1 = \{ v_i v_j ~|~ t_i | t_j \}$ and $\overrightarrow{E_2} = \{ v_i \rightarrow v_j ~|~ (t_i \subseteq t_j) \mbox{ and } i<j \}$
 define a cochordal-mixed partition of the edges of $\overline{G}$.
 Since $R$ is a covering subtree of $\mathcal{T}$, for all $1 \le i,j \le n$, $t_i \cap R$ is a subtree of $R$ and $t_i \cap t_j \cap R = \emptyset$ if and only if $t_i \cap t_j = \emptyset$.
Therefore, subtrees $\{ t_i \cap R ~|~ t_i \in \mathcal{T} \}$ of $R$ form an $R$-cochordal representation of the graph $(V, E_1)$.
 
\textbf{\ref{3mixed} $\Rightarrow$ \ref{3bushy}}: 
This part of the proof combines elements of the proof that complements of cointerval-mixed graphs are interval filament graphs \cite{gavril2000} and the proof that subtree filament graphs are subtree overlap graphs \cite{Jess}. 
Let $G=(V,E)$ be the complement of an $R$-cochordal-mixed graph.
Let $E_1$ and $E_2$ be a partition of the edges of $\overline{G}$ 
such that $(V, E_1)$ is an $R$-cochordal graph
and $ (V, \overrightarrow{E_2})$ is a transitive orientation of $(V, E_2)$ such that
for all $1 \le i,j,k \le n$, if
$(v_i \rightarrow v_j) \in \overrightarrow{E_2}$ and $v_j v_k \in E_1$ then $v_i v_k \in E_1$. 
Let subtrees $\mathcal{T} = \{ t_1, \ldots, t_n \}$ of tree $R = (V_R, E_R)$ be an $R$-cochordal representation of $(V, E_1)$, that is, for all $1 \le i,j \le n$, $v_i v_j \in E_1$ if and only if $t_i | t_j$.
      
Suppose that 
$(v_i \rightarrow v_j) \in \overrightarrow{E_2}$. Then $t_i \cap t_j \ne \emptyset$. 
Furthermore, if $t_i \not\subseteq t_j$, then
replacing $t_i$ with $t_i \cap t_j$ produces another $R$-cochordal representation of $(V, E_1)$, as justified by the following argument from \cite{gavril2000}.
Suppose that $t_i \not\subseteq t_j$.
By the definition of $(V, \overrightarrow{E_2})$, for all $1 \le k \le n$, if $v_i v_k \notin E_1$ then $v_j v_k \notin E_1$. Equivalently, every $t_k \in \mathcal{T}$ that intersects $t_i$ also intersects $t_j$ (as well as $t_i \cap t_j$ since $t_i$, $t_j$, and $t_k$ are all subtrees of a tree).
Therefore, replacing $t_i$ with $t_i \cap t_j$ produces another $R$-cochordal representation of $(V, E_1)$.
Applied repeatedly, this transformation results in an $R$-cochordal representation of $(V, E_1)$ such that for all $1 \le i,j \le n$, $(v_i \rightarrow v_j) \in \overrightarrow{E_2}$ implies $t_i \subseteq t_j$. 
      
Let $T'$ be the tree $R$ with $n$ additional nodes: $x_1, \ldots, x_n \notin V_R$ where, for each $1 \le i \le n$, $x_i$ is adjacent in $T'$ to exactly one arbitrary node of $t_i \cap R$.
Then, for $1 \le i \le n$, let $t'_i = t_i \cup \{x_i\} \cup \{x_k ~|~ (v_k \rightarrow v_i) \in \overrightarrow{E_2} \}$ and let $\mathcal T' = \{ t'_1, \ldots, t'_n \}$. 
Each element of $\mathcal T'$ induces a connected subgraph and therefore a subtree of $T'$.
Suppose there are two elements of $\mathcal T'$, $t'_i$ and $t'_j$ such that $i \ne j$ and $t'_i = t'_j$. Then $x_i$ is in $t'_j$ and $x_j$ is in $t'_i$ which implies that $(v_i \rightarrow v_j)$ and  $(v_j \rightarrow v_i)$ are both in $ \overrightarrow{E_2}$. But this contradicts transitivity since 
$\overline{G}$ is a simple graph. Therefore,
the elements of $\mathcal T'$ are distinct.

For each $1 \le i,j \le n$, $i \ne j$, $v_i v_j$ is in exactly one of $E_1$, $E_2$, or $E$.
Since $\mathcal{T}$ is a cochordal representation of $(V,E_1)$, $v_i v_j \in E_1$ if and only if $t_i | t_j$.
Of the nodes $x_1, \ldots, x_n$, only $x_i$ (respectively $x_j$) and those corresponding to subtrees contained in or equal to $t_i$ (respectively $t_j$) are in $t'_i$ (respectively $t'_j$). Furthermore, $t_i \subset t'_i$ and $t_j \subset t'_j$. Therefore $t_i | t_j$ if and only if  $t'_i | t'_j$ and so $v_i v_j \in E_1$ if and only if $t'_i | t'_j$.
If $v_i v_j \in E_2$ then suppose without loss of generality that $(v_i \rightarrow v_j) \in \overrightarrow{E_2}$. Then, by our earlier argument, $t_i \subseteq t_j$. In addition, every vertex of $t'_i \backslash t_i$ is also in $t'_j$ by transitivity of  $\overrightarrow{E_2}$, and $x_j \in t'_j \backslash t'_i$.
Therefore, 
if $(v_i \rightarrow v_j) \in \overrightarrow{E_2}$ then $t'_i \subset t'_j$, and
if $v_i v_j \in E_2$ then $t'_i \subset t'_j$ or $t'_j \subset t'_i$.
Finally, if $v_i v_j \in E$ then,
since $\mathcal{T}$ is a cochordal representation of $(V,E_1)$, $t_i \cap t_j \ne \emptyset$ and therefore 
$t'_i \cap t'_j \ne \emptyset$. In addition, $x_i \in t'_i \backslash t'_j$ and $x_j \in t'_j \backslash t'_i$.
Thus,  $v_i v_j \in E$ implies $t'_i \between t'_j$. We conclude that the subtrees $\mathcal T'$ of tree $T'$ form a subtree overlap representation of $G$.

Furthermore, since for each $1 \le i \le n$, $t_i \subset t'_i$ and each $x_i$ is adjacent to a vertex of $R$, subtrees $\mathcal{T}'$ of tree $T'$ form an $R$-covered subtree overlap representation of $G$ in which $R$ is bushy.

\textbf{\ref{3bushy} $\Rightarrow$ \ref{3covered}}: Obvious.
\end{proof}

The classes of $K_1$-covered subtree overlap graphs and bushy $K_1$-covered subtree overlap graphs are equivalent to the class of cocomparability graphs. 
This follows from \cite{GolSch} combined with the observation that subtrees of a tree that all have a vertex in common overlap if and only if neither is contained in the other.
Thus, Theorem \ref{th:3parts}
generalises characterisations of cocomparability graphs as the overlap graphs of subtrees of a tree where all subtrees have a vertex in common, the complements of cochordal-mixed graphs where all edges are in the $E_2$ block of the partition, and the overlap graphs of subtrees of a star \cite{eowynStewart, gavril2000, GolSch}.

\begin{theorem}\label{bigThm}
Let $G$ be a graph and $\mathcal{S} \ne \{K_1\}$ be a nonempty set of trees that is closed under edge subdivision. The following statements are equivalent:
\begin{enumerate}
\item $G$ is an $\mathcal{S}$-covered subtree overlap graph. \label{covered}
\item $G$ is the complement of an $\mathcal{S}$-cochordal-mixed graph. \label{mixed}
\item $G$ is a bushy $\mathcal{S}$-covered subtree overlap graph. \label{bushy}
\item $G$ is an $\mathcal{S}$-subtree-filament graph. \label{filament}
\end{enumerate}
\end{theorem}
\begin{proof}

\textbf{\ref{covered} $\Leftrightarrow$ \ref{mixed} $\Leftrightarrow$ \ref{bushy}}: by Theorem \ref{th:3parts}.

\textbf{\ref{filament} $\Leftrightarrow$ \ref{mixed}}: 
By Theorem 4 of \cite{gavril2000}, a graph is a subtree filament graph if and only if it is the complement of a cochordal-mixed graph. 
In the proof of that theorem, a subtree-filament representation of a graph $G=(V,E)$ is transformed to a cochordal representation of the graph $(V,E_1)$ on the same host tree, where $E_1$ and $E_2$ is a cochordal-mixed partition of the edges of $\overline{G}$. Thus \ref{filament} implies \ref{mixed}. 

For the other direction, suppose that $G$ is the complement of an $R$-cochordal-mixed graph where $R \in \mathcal S$, and let $E_1$ and $E_2$ be a cochordal-mixed partition of the edges of $\overline{G}$. By Lemma \ref{lem:bushyGivesNice}, there is a $T$-cochordal representation of $(V,E_1)$ that satisfies Property \ref{A}, such that $T$ is a subdivision of $R$. The proof of Theorem 4 of \cite{gavril2000} transforms a $T$-cochordal representation of $(V,E_1)$ and cochordal-mixed partition $E_1$ and $E_2$ of the edges of $\overline{G}$ to a $T$-subtree-filament representation of $G$, provided that the 
$T$-cochordal representation of $(V,E_1)$ satisfies Property \ref{A}.
Therefore, since $T \in \mathcal S$, we have \ref{mixed} implies \ref{filament}.
\end{proof}

Theorem \ref{bigThm} does not hold for $\mathcal{S} = \{K_1\}$ since only complete graphs are 
$K_1$-subtree-filament graphs while, as previously noted, 
the classes of $K_1$-covered subtree overlap graphs, 
$K_1$-cochordal-mixed graphs, and
bushy $K_1$-covered subtree overlap graphs
are all equivalent to cocomparability graphs.

When $\mathcal{S}$ is the set of subdivisions of $K_2$, Theorem \ref{bigThm} becomes the following characterisation of interval filament graphs, which includes results of \cite{CGO} and \cite{gavril2000}.

\begin{corollary}
The following statements are equivalent for a graph $G$:
$G$ is a path-covered subtree overlap graph;
$G$ is the complement of a cointerval-mixed graph;
$G$ is the overlap graph of subtrees of a caterpillar;
$G$ is an interval filament graph.
\end{corollary}

\section{Conclusion}

We have presented two main results: 

\begin{enumerate}
\item
The following graph classes are equivalent for any tree $R$: $R$-covered subtree overlap graphs, the complements of $R$-cochordal-mixed graphs, and bushy $R$-covered subtree overlap graphs.
\item
The following graph classes are equivalent for any nonempty set of trees $\mathcal{S} \ne \{K_1\}$ that is closed under edge subdivision:
$\mathcal{S}$-covered subtree overlap graphs,
the complements of $\mathcal{S}$-cochordal-mixed graphs,
bushy $\mathcal{S}$-covered subtree overlap graphs,
and $\mathcal{S}$-subtree-filament graphs.
\end{enumerate}

The first result is a generalization of characterisations of cocomparability graphs, as can be seen in the simplest case of Theorem \ref{th:3parts}, when $R = K_1$. 
The second result generalises characterisations of interval filament graphs. The simplest case of Theorem \ref{bigThm}, when $\mathcal{S}$ is the set of subdivisions of $K_2$, states that the following graph classes are equivalent: path-covered subtree overlap graphs, the complements of cointerval-mixed graphs, the overlap graphs of subtrees of caterpillars, and interval filament graphs.
The second result suggests that the $\mathcal{S}$-covered subtree overlap graphs, for sets $\mathcal{S}$ of trees closed under edge subdivision, might be a useful way of breaking down the class of subtree overlap graphs. We propose three avenues based on that idea for future study.

While some subclasses of subtree overlap graphs can be recognised in polynomial time (including interval, permutation, cocomparability, chordal, circular arc, and circle graphs), for others the recognition problem is NP-complete (including interval filament graphs \cite{Perm},  overlap graphs of subtrees of a tree with a bounded number of leaves, the overlap graphs of subtrees of subdivisions of a fixed tree with at least three leaves, and the overlap graphs of paths in a tree with bounded maximum degree \cite{JessThesis,PergelThesis}). The complexity of the recognition problem for subtree overlap graphs is open.  
An efficient recognition algorithm that could output subtree overlap representations for yes instances would have significant algorithmic implications since several optimisation problems that are NP-complete in general can be solved efficiently for subtree overlap graphs when a subtree overlap representation is given \cite{eowynStewart, gavril2007, gavril2009, gavril2011, Keil}. Does the recognition problem on $\mathcal{S}$-covered subtree overlap graphs give insight into the recognition problem on subtree overlap graphs as a whole?

Several optimisation problems remain NP-hard on subtree overlap graphs by virtue of hardness results on the subclasses. It would be interesting to explore the possible P vs. NP-complete boundaries for various optimisation problems within the 
$\mathcal{S}$-covered subtree overlap graphs
over sets $\mathcal{S}$ of trees that are closed under edge subdivision.

Parameters of chordal graphs based on their subtree intersection representations include leafage, that is,
the minimum number of leaves in the host tree of a representation \cite{LinMcKeeWest}, and vertex leafage, that is,
the minimum maximum number of leaves of a subtree in a representation \cite{ChaplickStacho}.
How do analogous and other parameters of subtree overlap graphs relate to the $\mathcal{S}$-covered subtree overlap graph classes of this paper?

\section{Acknowledgements}
The authors gratefully acknowledge support from an NSERC Discovery grant, an iCORE ICT Graduate Student Scholarship and a University of Alberta Dissertation Fellowship.

\bibliographystyle{plain}

\end{document}